\documentclass[journal]{IEEEtran} 
\usepackage{graphicx}      \usepackage{times} \usepackage{amsmath} \usepackage{amssymb}  \usepackage{amsfonts}
\usepackage{amsthm}
\usepackage{url}
\usepackage{tabularx}
\usepackage{booktabs}
\usepackage{caption}
\usepackage{subcaption}
\usepackage{color}
\usepackage[english]{babel}
\usepackage{cite}

\newcommand{\rev}[1]{#1}

\newtheorem{theorem}{Theorem}[section]

\newtheorem{proposition}[theorem]{Proposition}
\newtheorem{corollary}[theorem]{Corollary}
\newtheorem{definition}{Definition}
\newtheorem{statement}{Statement}
\newtheorem{remark}{Remark}

\newcommand{\ub}{\overline{u}}
\newcommand{\round}[1]{\operatorname{\rho}\left({#1}\right)}
\newcommand{\integer}[1]{\operatorname{int}\left({#1}\right)}
\newcommand{\fractional}[1]{\operatorname{frac}\left({#1}\right)}
\newcommand{\sign}[1]{\operatorname{sign}\left({#1}\right)}

\begin{document}
\title{Switched control for quantized feedback systems:
invariance and limit cycle analysis}

\author{Alessandro~Vittorio~Papadopoulos,~\IEEEmembership{Member,~IEEE,}
        Federico~Terraneo,\\
        Alberto~Leva,~\IEEEmembership{Member,~IEEE,}
        Maria~Prandini,~\IEEEmembership{Senior Member,~IEEE}\thanks{A.V. Papadopoulos is with M{\"a}lardalen University, V{\"a}ster{\aa}s, Sweden, (e-mail: alessandro.papadopoulos@mdh.se), and F. Terraneo, A. Leva, and M. Prandini are with the Dipartimento di Elettronica, Informazione e Bioingegneria, Politecnico di Milano, Milano, 20133, Italy, e-mail: \{federico.terraneo, alberto.leva, maria.prandini\}@polimi.it).}\thanks{This work was done when the first author was a post-doctoral researcher at Politecnico di Milano, and is supported by the European Commission under the project UnCoVerCPS with grant number 643921. }
}

\maketitle

\begin{abstract}
We study feedback control for discrete-time linear time-invariant systems in the presence of quantization both in the control action and in the measurement of the controlled variable.
While in some application the quantization effects can be neglected, when high-precision control is needed, \rev{they have to be} explicitly accounted for in control design. In this paper we propose a switched control solution for minimizing the effect of quantization of both the control and controlled variables in the case of a simple integrator with unitary delay, a model that is quite common in the computing systems domain\rev{, for example in thread scheduling, clock synchronization, and resource allocation}.
We show that the switched solution outperforms the one without switching, designed by neglecting quantization, and analyze necessary and sufficient conditions for the controlled system to exhibit periodic solutions in the presence of an additive constant disturbance affecting the control input.
Simulation results provide evidence of the effectiveness of the approach.
 \end{abstract}


\IEEEpeerreviewmaketitle

\section{Introduction}
\label{sec:IntroAndMotivation}

This paper deals with quantized feedback control for discrete-time linear time-invariant control systems. In particular, we consider the effect of quantization of both the measurements and the control actions.

In general, any digital implementation of a control system entails input and output quantization.  This is typically the case when the output measurements  used for feedback and the control actions applied to the controlled process are transmitted via a digital communication channel, \rev{\cite{Mitter2004,Baillieul2002}.}
Depending on the specific application, quantization effects can become relevant and significantly affect the control system performance. While in some applications the quantization effects can be neglected, when high-precision control is needed, quantization has to be explicitly accounted for in control design.

\rev{Given a system that is stabilized by a standard linear time-invariant feedback controller when there is no quantization, the problem addressed herein is to find a switched controller that steers the system towards the smallest possible invariant set that includes the origin when its control input and output are quantized. We focus, in particular, on a discrete time linear system described by an integrator with a one time-unit delay. The system is affected by an additive constant bias on the control input, and both control input and controlled output measurements are quantized via a rounding operator.}

\rev{Despite its simplicity, this system structure appears in several problems pertaining to the domain of computing systems. For example it represents the dynamics from reservation to cumulative CPU time in task scheduling, a typical source of disturbance being the latency of the preemption interrupt~\cite{bib:LevaMaggio-2010a,papadopoulos2015rtsj}. It models the disturbance to error dynamics in clock synchronization for wireless sensor networks, where the most relevant source of disturbance is given by temperature variations in the oscillator crystals~\cite{leva2015tcst}. It plays a role in server systems~\cite{kihl+web07}, queuing systems~\cite{KjaerKihlRobertsson_CDC2007}, and so forth, as can be observed from the variety of problems mentioned in~\cite{bib:DiaoEtAl-2005a,bib:HellersteinEtAl-2004a,papadopoulos2015mcmds,astrom2010feedback}. Quantizers are present in virtually the totality of these applications, and dealing with their effect is important when high-performance is required.
In fact, several of the problems just listed require zero error in the presence of constant inputs, hence the relevance of quantization becomes apparent. Constant (or practically constant) are for example thermal disturbances experienced by wireless nodes in a climatized environment. In such an application context, temperature variations are very small and slow, because they are smoothed by the environment thermal dynamics and counteracted by temperature control, and abrupt variations may occur but only sporadically, for example when turning the air conditioners on once per day or week. }

\rev{The considered linear system is stabilizable, and in the absence of quantization, one can introduce a standard proportional-integral (PI) controller to compensate for a constant load disturbance and bring the state trajectories to the zero equilibrium. The presence of input and output quantizers degrades the
PI controller performance, introducing oscillations in the quantized output with an excursion that is equal to twice the quantizer resolution. Such oscillations may be not admissible when dealing with high precision computing systems. Our goal in this paper is to design a better performing controller, while maintaining a PI-like structure in order to ease implementation and tuning. Invariant set and reachability analysis are the methods adopted to assess the properties of the designed control scheme.  }

More precisely, we propose a switched variant of the PI controller to address quantization and minimize its effect on the feedback control system performance. We then show that when the disturbance is constant,  the switched control solution presents an invariant set for the quantized control input and output variables such that the quantized output is either zero or has a unitary amplitude \rev{(corresponding to the least significant bit, hence to the minimum representable quantity)}. A numerical reachability analysis study shows that, if the PI controller is suitably tuned, this invariant set is a global attractor. Necessary and sufficient conditions for the existence of a periodic solution in the (unquantized) control input and output variables are given as well.

\rev{Many papers in the literature address control of quantized linear systems. 
Most of them focus on stabilization at the  zero equilibrium in absence of disturbances.
Contributions can be classified based on the characteristics of the adopted quantizer. If the quantizer has a finite resolution, like in our paper where uniform quantization is adopted, then, \cite{delchamps1990tac} shows that classical stability cannot be achieved and introduces the \emph{practical stability}
notion for quantized systems. More specifically, \cite{delchamps1990tac} proves that,
given an unstable discrete time system that is stabilizable, if the state measurements are quantized, then, there is no control strategy that makes all trajectories of the quantized state-feedback system asymptotically converge to zero, and only convergence to an invariant set around zero can be obtained.
Classical results on asymptotic stability of the origin are recovered in \cite{BL2000, KD2008} by changing the resolution of the quantizer depending on the state behavior, and hence making the resolution higher and higher while approaching the origin. This approach has been extended to input to state and $l_2$ stabilization in presence of a disturbance input in \cite{LN2007} and \cite{KD2008}, respectively.  
When a logarithmic quantizer with (countably) infinite  quantization levels is adopted,
the resolution of the quantizer is infinite close to the origin, and global asymptotic stability
can be achieved, \cite{Elia2001,Coutinho2010}. However, when finite-level logarithmic quantizers
are used, practical stability results can only be proven. Analysis of practical stability and constructive results on how to design finite-level logarithmic state quantizers guaranteeing practical stability are given in, e.g., \cite{Elia2001,Maestrelli2015}.  \\
It is worth noticing that most papers in the literature consider quantization of either the control input (see, e.g., \cite{PLPB02,Fu2005,Tian2008}) or the controlled output (see, e.g., \cite{chou1996jsme,delchamps1990tac,BL2000,KD2008,feng1997tac,raisch1995hs,CO2008,ishido2010acc,ishido2011scl}), whereas only a few address the set-up considered in this paper, where both control input and controlled output are quantized. This is the case in \cite{Cepeda2008,Picasso2007,Coutinho2010}.
Whereas logarithmic quantizers with infinite quantization levels are considered in  \cite{Coutinho2010}, in \cite{Cepeda2008}, input and output quantizers are assumed to have a finite number of quantization levels. Practical stabilization of a double integrator system is studied in \cite{Cepeda2008}, showing how the parameters defining the quantizers should be set for the practical stability result to hold. Extension to higher order integrator models is outlined as well, focusing however on stabilization without disturbances acting on the system.
The work closest to the present paper is \cite{Picasso2007}, where \emph{pre-defined} finite resolution quantizers on both input and output are given and a feedback controller is designed to achieve some control goal. More precisely, in \cite{Picasso2007}, practical stabilization of unstable discrete time linear systems is addressed, and a quantized static state-feedback controller is designed that brings the state of the system to some invariant set around the origin in a finite number of steps. Our approach differs from \cite{Picasso2007} in that we address disturbance compensation, and we introduce a switched output-feedback controller to make the state of the controlled system reach an invariant set around the origin. Disturbance compensation and dynamic state/output-feedback control are not addressed in \cite{Picasso2007} and related work. In turn, while the methodology in \cite{Picasso2007} is of general applicability, our design is tailored to a simple system model and not easily extendable to different higher dimensional models. }

The rest of the paper is organized as follows. Section~\ref{sec:FLOPSYNC} first describes the control scheme without switching, and highlights how quantization deteriorates the performance of the control system. The switched solution that allows for minimizing the effect of quantization is then presented in the same section. Section~\ref{sec:QcontrolInvariance} provides necessary and sufficient conditions for entering the invariant set. A numerical reachability analysis study is performed in Section~\ref{sec:Reachability} for identifying the controller parameter tuning that makes such an invariant set a global attractor. Section~\ref{sec:limitCycles} gives necessary and sufficient conditions for the existence of periodic solutions. Finally, Section~\ref{sec:Experiments} provides evidence of the effectiveness of the approach via a simulation study, while Section~\ref{sec:Conclusions} concludes the paper.

 \section{Basic control scheme and its switched variant}
\label{sec:FLOPSYNC}

\subsection{Notation}

We now introduce some notation that will be used in the paper developments.
\begin{definition}[Sign function]
The sign function of a real number $z$ is defined as:
\begin{align*}
\sign{z} :=
\begin{cases}
1, & z > 0\\
0, & z = 0\\
-1, & z < 0
\end{cases}
\end{align*}
\end{definition}

\begin{definition}[Integer part of a number]
The integer part of a real number $z$ is defined as:
\begin{align*}
\integer{z} :=
\begin{cases}
\lfloor z \rfloor, & z\geq 0\\
\lceil z \rceil, & z < 0
\end{cases}
\end{align*}
where $\lfloor z \rfloor$ is the largest signed integer smaller than or equal to $z$ and  
$\lceil z \rceil$ is the smaller  signed  integer larger than or equal to $z$.            
\end{definition}  

\begin{definition}[Fractional part of a number]
The fractional part of a real number $z$ is defined as:
\begin{align*}
\fractional{z}: = z - \integer{z}
\end{align*}
\end{definition}

A quantizer maps a real-valued function into a piecewise constant function taking values in a discrete set, and here it is defined as the rounding operator.

\begin{definition}[Rounding operator]
Given a real number $z$, its rounding $\rho: \mathbb{R} \to \mathbb{Z}$ is defined as:
\begin{align*}
\round{z} :=
\begin{cases}
\sign{z} \cdot \vert \integer{z} \vert, & 0 \leq |\fractional{z}| < \frac{1}{2}\\
\sign{z} \cdot \left( \vert \integer{z} \vert + 1\right), & \frac{1}{2} \leq |\fractional{z}| < 1
\end{cases}
\end{align*}
\end{definition}

\begin{definition}[Rounding error]
Given a real number $z$, its rounding error is:
\begin{align*}
\Delta_z := z - \round{z}.
\end{align*}
\end{definition}
Notice that according to the provided definitions, the rounding error of a real number $z$ is always bounded as $|\Delta_z| \leq \frac{1}{2}$.

Finally, note that given two real numbers $a\in \mathbb{R}$, and $b \in \mathbb{R}$, we have that $\round{\round{a}+b} = \round{a} + \round{b}$.

\subsection{The basic scheme}

We consider a system with control input $u$ and output $e$, which is governed by the following equation
\begin{equation}
 e(k+1) = e(k)+\round{u(k)}+d(k),
 \label{eq:errorDynamics}
\end{equation}
where $d$ is some additive \rev{constant yet unknown} disturbance on the quantized control action $\round{u}$.

The output $e$ represents some error signal and should be driven to zero by compensating the disturbance $d$ through the control input $u$. To this purpose, quantized measurements of $e$ are available for feedback.
Due to the quantization of both $u$ and $e$, the disturbance might not be exactly compensated and the goal is to design an output feedback compensator so that $e$ is kept below the minimum resolution as defined by the quantizer ($\round{e}=0$).

The transfer function between the \emph{residual disturbance} $\round{u}+d$ and the controlled variable $e$ is given by
\begin{equation}
 P(z) = \frac{1}{z-1},
 \label{eqn:Background-edyn-u-TF}
\end{equation}
which is a discrete time integrator with a one time unit delay.

Suppose that disturbance $d$ is constant, and neglect the quantization for the time being.
Then, a discrete-time Proportional Integral (PI) controller described via the transfer function:
\begin{equation}
 R(z)=\frac{1-\alpha z}{z-1},
  \label{eqn:Background-R}
\end{equation}
would suffice to drive $e$ to zero with a rate of convergence that can be set via the parameter $\alpha$.
Indeed, if we neglect the quantizers, the effect of the disturbance $d$ on the output $e$ can be described via the (closed-loop) transfer function
\begin{equation*}
F(z)= \frac{P(z)}{1-R(z)P(z)} = \frac{z-1}{z(z+\alpha-2)},
\end{equation*}
which corresponds to an asymptotically stable linear system if  $1 < \alpha < 3$.
Hence, in the absence of quantization effects, the PI controller guarantees that the error converges to zero in the presence of a constant disturbance, with a rate of convergence that depends on the parameter $\alpha$. If $\alpha=2$, output $e$ would be brought to zero in  two time units.

Figure~\ref{fig:flopsyncControlScheme} shows the resulting control scheme, including the quantizers.
\begin{figure}[t]
\centering
\includegraphics[scale=1]{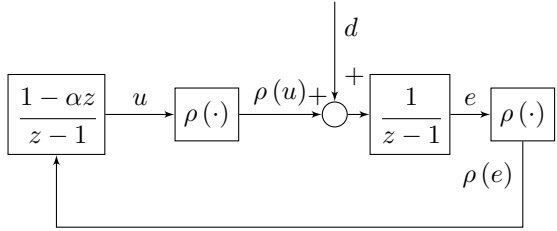}
\caption{Basic control scheme with quantizers.}
\label{fig:flopsyncControlScheme}
\end{figure}

\subsection{The effect of quantization}
As anticipated in the introduction, whenever high-precision control is needed, quantization can significantly deteriorate the performance of the control system. Indeed, quantization effects are not negligible in almost all the applications where a digital implementation is in place.

\rev{In particular, in the case of the scheme in Figure~\ref{fig:flopsyncControlScheme}, a constant disturbance may cause the system to end up in a limit cycle where the excursion in amplitude of the quantized error is $2$. An example is shown in Figure~\ref{fig:FLOPSYNC-res}, with $\alpha = 1.4$, $d(k) = \overline{d} = 1.2$, and the control system initialized as $e(0) = 2$, $u(0) = 0$.
This figure, and, more precisely, the behavior of the error signal $e$, shows that the system with transfer function $P(z)$ integrates over time the residual between the disturbance $\overline d$ and the quantized control input $\round{u}$. Due to the quantization on the system output $e$, the PI controller keeps its control action constant as long as $\round{e}$ is zero. It then reacts when the integrated residual disturbance exceeds the threshold $1/2$ in amplitude and makes the quantized output $\round{e}$ change value from 0 to either $1$ or $-1$, depending on its sign. The control signal reverses the sign of the residual disturbance, thus causing the quantized output $\round{e}$ too to change sign. As a result, $\round{e}$ is brought to a limit cycle where it keeps commuting between $-1$ and $1$, with an excursion in amplitude that is equal to $2$.}

\begin{figure}[tb]
\centering
\includegraphics[scale=1]{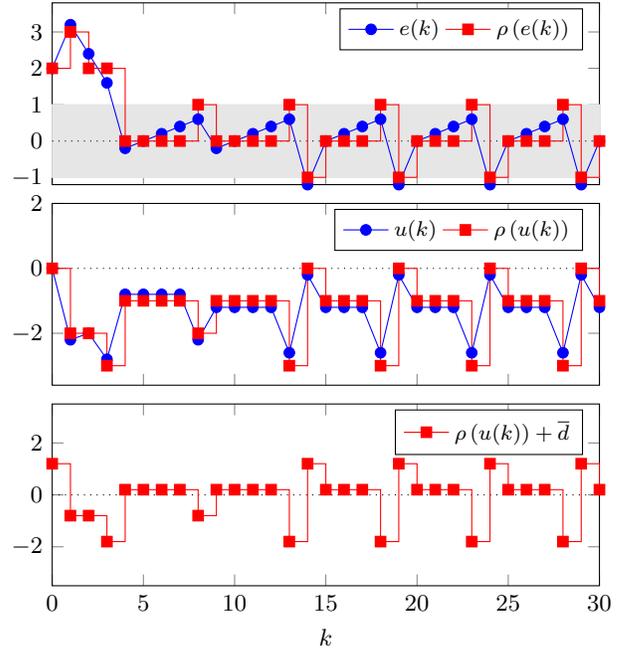}
\caption{\rev{The effect of quantization in the control scheme in Figure~\ref{fig:flopsyncControlScheme}: error signal $e$ and its quantized version (top plot), control input $u$ and its quantized version (middle plot), and residual disturbance $\round{u}+\overline{d}$ (bottom plot).}}
\label{fig:FLOPSYNC-res}
\end{figure}

\subsection{The proposed switched control scheme}
\label{sec:ProblemStatement}
In this section, we propose a switched control scheme that reduces the effect of quantization, steering the system to a limit cycle of an amplitude that is half of the one obtained with the control scheme in~Figure~\ref{fig:flopsyncControlScheme}. The proposed solution has the advantage of still adopting simple controllers, which leads to a system easily implementable in an embedded device, with very low overhead.

The controller is composed of a linear part with transfer function
\begin{align*}
\tilde{R}(z) = \dfrac{\alpha z - 1}{z}.
\end{align*}
and a switched part where the control action $\tilde{u}$ computed by $\tilde{R}(z)$ is set as the input to the following modified integrator:
\begin{align*}
\begin{cases}
u(k+1) = u(k) + \tilde{u}(k+1), & \text{if } \round{e(k+1)} \neq 0\\
u(k+1) = \round{u(k)} + \tilde{u}(k+1), & \text{if } \round{e(k+1)} = 0
\end{cases}
\end{align*}
that finally computes the actual control input $u$, based on the quantized  error measurements $\round{e}$.
Figure~\ref{fig:switchedControlScheme} shows the resulting switched control scheme.
\begin{figure}[t]
\centering
\includegraphics[scale=1]{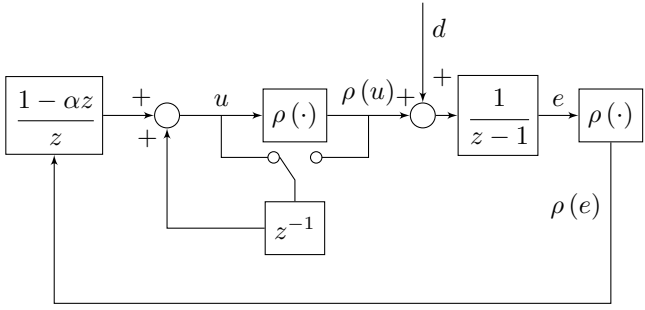}
\caption{Proposed switched control scheme.}
\label{fig:switchedControlScheme}
\end{figure}

Note that if $\round{e(k+1)} \neq 0$, then, the effect of  $\round{e}$ on $u$ is describe by the transfer function $R(z)$ of the PI controller previously presented. Furthermore, in the absence of quantization, the two schemes in Figures~\ref{fig:flopsyncControlScheme} and~\ref{fig:switchedControlScheme} coincide.

The switched control system dynamics is characterized by the state variables $u$ and $e$, and can be expressed as follows:
\begin{itemize}
\item if $\round{e(k+1)} = \round{e(k) + \round{u(k)} + d(k)} = 0$, then:
\begin{align}
&\begin{cases}
e(k+1) = e(k) + \round{u(k)} + d(k)\\
u(k+1) = \round{u(k)} + \round{e(k)}
\end{cases}
\label{eq:modeUguale0}
\end{align}
\item if $\round{e(k+1)} = \round{e(k) + \round{u(k)} + d(k)} \neq 0$, then:
\begin{align}
&\begin{cases}
\begin{aligned}
e(k+1) &= e(k) + \round{u(k)} + d(k)\\
u(k+1) &= u(k) + \round{e(k)} \\
       &~ - \alpha \round{e(k) + \round{u(k)} + d(k)}
\end{aligned}
\end{cases}
\label{eq:modeDiverso0}
\end{align}
\end{itemize}

 \section{Invariant set analysis}
\label{sec:QcontrolInvariance}
In this section we prove that, for a constant disturbance $d(k)=\overline{d}$, the proposed control scheme admits an invariant set in the quantized state variables $\round{e}$ and $\round{u}$, and within that set the amplitude of the quantized  error oscillations is $1$.

We characterize the conditions under which the control system enters this invariant set. To this purpose it is convenient to express the control input as the quantized disturbance compensation term $-\round{\overline{d}}$ plus the residual:
\begin{align}
&u(k) = -\round{\overline{d}} + \ub(k),
\label{eq:utilde}
\end{align}
and let
\begin{align}
&\Delta_d = \overline{d} - \round{\overline{d}},
\label{eq:deltad}
\end{align}
be the rounding error of the disturbance. We can then rewrite the control system dynamics in the state variables $e$ and $\ub$ as:
\begin{itemize}
\item if $\round{e(k+1)} = \round{e(k) + \round{\ub(k)} + \Delta_d} = 0$, then:
\begin{align}
&\begin{cases}
e(k+1) = e(k) + \round{\ub(k)} + \Delta_d\\
\ub(k+1) = \round{\ub(k)} + \round{e(k)}
\end{cases}
\label{eq:modeUguale0_delta}
\end{align}
\item if $\round{e(k+1)} = \round{e(k) + \round{\ub(k)} + \Delta_{d}} \neq 0$, then:
\begin{align}
&\begin{cases}
\begin{aligned}
e(k+1) &= e(k) + \round{\ub(k)} + \Delta_{d}\\
\ub(k+1) &= \ub(k) + \round{e(k)}\\
               &~ - \alpha \round{e(k) + \round{\ub(k)} + \Delta_d}
\end{aligned}
\end{cases}
\label{eq:modeDiverso0_delta}
\end{align}
which better shows that the rounding error of the disturbance is integrated by the process dynamics.
\end{itemize}

\begin{theorem} \label{th:qinv}
Let $1 < \alpha < \frac{3}{2}$, and consider the system described by~\eqref{eq:modeUguale0_delta} and~\eqref{eq:modeDiverso0_delta}. If, at some time $k$
\begin{subequations}
\label{eq:hypotheses}
\begin{align}
&-\frac{1}{2} < e(k) < \frac{1}{2} \label{eq:hypotheses_a}\\
&1 \leq \alpha -\ub(k)\sign{\Delta_d} < \frac{3}{2} \label{eq:hypotheses_b}\\
&-\frac{1}{2} < \ub(k) < \frac{1}{2} \label{eq:hypotheses_c}
\end{align}
\end{subequations}
then, for all the subsequent time steps $k+h$, $h>0$:
\begin{align}
(\round{e(k+h)},&\round{\ub(k+h)}) \in \nonumber\\
&\left\lbrace (0,0), (\sign{\Delta_d},-\sign{\Delta_d}) \right\rbrace.  \label{eq:invariantSet}
\end{align}
Moreover, $\left\lbrace (0,0), (\sign{\Delta_d},-\sign{\Delta_d}) \right\rbrace$ is the smallest invariant set for $\round{e}$ and $\round{\ub}$, when the system evolves starting from~\eqref{eq:hypotheses}.
\end{theorem}
\begin{proof}
Let us first consider the case where $\Delta_d = 0$. Given the error evolution in~\eqref{eq:modeUguale0_delta}-\eqref{eq:modeDiverso0_delta}, we get from~\eqref{eq:hypotheses} that:
\begin{align*}
& e(k+1) = e(k) + \round{\ub(k)} + \Delta_{d} = e(k).
\end{align*}
Then $\round{e(k+1)} = \round{e(k)} = 0$, and by~\eqref{eq:hypotheses_a} the system evolves according to~\eqref{eq:modeUguale0_delta}:
\begin{align}
&\begin{cases}
\begin{aligned}
e(k+1) &= e(k) + \round{\ub(k)} + \Delta_{d}\\
\ub(k+1) &= \round{\ub(k)} + \round{e(k)}
\end{aligned}
\end{cases} \quad \Rightarrow \nonumber\\
&\begin{cases}
e(k+1) = e(k)\\
\ub(k+1) = 0
\end{cases}
\label{eq:caseDelta0}
\end{align}
The first equation satisfies~\eqref{eq:hypotheses_a}, and the second equation satisfies both~\eqref{eq:hypotheses_b} and~\eqref{eq:hypotheses_c}, so that the corresponding system keeps evolving according to~\eqref{eq:caseDelta0}.
In addition, $\left(\round{e(k+1)},\round{\ub(k+1)}\right)$ is equal to $\left(0,0\right)$, and the system will keep staying in $(0,0)$ for all time $k + h$, with $h>0$. This concludes the proof for the case when $\Delta_d = 0$.

We now consider the case when $0 < \Delta_d \leq 1/2$. Derivations for the case $-1/2 \leq \Delta_d < 0$ are analogous, and hence omitted. Given~\eqref{eq:hypotheses_c}, we have:
\begin{align*}
e(k+1) = e(k) + \round{\ub(k)} + \Delta_{d} &= e(k) + \Delta_{d}.
\end{align*}
Since $-1/2 < e(k) < 1/2$ in~\eqref{eq:hypotheses_a}, and $0 < \Delta_d \leq 1/2$, then
\begin{align*}
-\frac{1}{2} < e(k) + \Delta_d < 1,
\end{align*}
and
\begin{align}
\round{e(k+1)} &= \round{e(k)+ \Delta_{d}} \nonumber \\
&=\begin{cases}
0, & \vert e(k) + \Delta_d \vert < \frac{1}{2}\\
1, & \frac{1}{2} \leq e(k) + \Delta_d < 1
\end{cases}
\label{eq:twoCases}
\end{align}
We can then distinguish the following two cases:
\begin{enumerate}
\item\label{it:case1} $\round{e(k+1)} = \round{e(k) + \round{\ub(k)} + \Delta_{d}} = 0$
\item\label{it:case2} $\round{e(k+1)} = \round{e(k) + \round{\ub(k)} + \Delta_{d}} = 1$
\end{enumerate}
\emph{Case~\ref{it:case1})}: The system evolves according to~\eqref{eq:modeUguale0_delta}:
\begin{align}
&\begin{cases}
\begin{aligned}
e(k+1) &= e(k) + \round{\ub(k)} + \Delta_{d}\\
\ub(k+1) &= \round{\ub(k)} + \round{e(k)}
\end{aligned}
\end{cases} \quad \Rightarrow \nonumber\\
&\begin{cases}
e(k+1) = e(k) + \Delta_d\\
\ub(k+1) = 0,
\end{cases} \label{eq:case1}
\end{align}
so that in one step the quantized state is brought to zero: $(\round{e(k+1)},\round{\ub(k+1)}) = (0,0)$. Since the first equation in~\eqref{eq:case1} satisfies~\eqref{eq:hypotheses_a}, and the second satisfies both~\eqref{eq:hypotheses_b} and~\eqref{eq:hypotheses_c}, we are back then to~\eqref{eq:twoCases}.

\emph{Case~\ref{it:case2})}: The system evolves according to~\eqref{eq:modeDiverso0_delta}:
\begin{align}
&\begin{cases}
\begin{aligned}
e(k+1) &= e(k) + \round{\ub(k)} + \Delta_{d}\\
\ub(k+1) &= \ub(k) + \round{e(k)}\\
                  &~ - \alpha \round{e(k) + \round{\ub(k)} + \Delta_d}
\end{aligned}
\end{cases} \quad \Rightarrow  \nonumber\\
&\begin{cases}
e(k+1) = e(k) + \Delta_d\\
\ub(k+1) = \ub(k) - \alpha
\end{cases} \label{eq:case2}
\end{align}

By~\eqref{eq:hypotheses_b}, we have:
\begin{align*}
&-\frac{3}{2} < \ub(k)-\alpha \leq -1,
\end{align*}
hence
\begin{align*}
\round{\ub(k+1)} = \round{\ub(k)-\alpha} = -1,
\end{align*}
so that $(\round{e(k+1)},\round{\ub(k+1)}) = (1,-1)$.

If we next compute:
\begin{align*}
e(k+2) &= e(k+1) + \round{\ub(k+1)} + \Delta_{d} \\
&= e(k+1) -1 + \Delta_{d},
\end{align*}
since $e(k+1) = e(k) + \Delta_d$, and in this case $1/2 \leq e(k) + \Delta_d < 1$:
\begin{align*}
&-\frac{1}{2} < e(k+1) - 1 + \Delta_d < \frac{1}{2},
\end{align*}
we then have
\begin{align*}
&\round{e(k+2)} = 0.
\end{align*}
The dynamics therefore evolves according to~\eqref{eq:modeUguale0_delta}, i.e.,
\begin{align}
&\begin{cases}
e(k+2) = e(k+1) + \round{\ub(k+1)} + \Delta_{d}\\
\ub(k+2) = \round{\ub(k+1)} + \round{e(k+1)}
\end{cases} \quad \Rightarrow \nonumber\\
&\begin{cases}
e(k+2) = e(k+1) - 1 + \Delta_d\\
\ub(k+2) = - 1 + 1 = 0
\end{cases} \label{eq:case2b}
\end{align}
so that $(\round{e(k+2)},\round{\ub(k+2)}) = (0,0)$. In $2$ steps the quantized state is brought to zero. The first equation in~\eqref{eq:case2b} satisfies hypothesis~\eqref{eq:hypotheses_a}, the second satisfies both~\eqref{eq:hypotheses_b} and~\eqref{eq:hypotheses_c}, hence we are back to~\eqref{eq:twoCases}.

All the above shows that starting from~\eqref{eq:hypotheses}, the system ends up evolving in the invariant set $\{(0,0),(1,-1)\}$ for $(\round{e},\round{\ub})$. Now we need to prove that this is the smallest invariant set.

Note that we have just shown that from~\eqref{eq:hypotheses} the system either enters the invariant set in $(0,0)$ or in $(1,-1)$, and in this latter case it evolves to $(0,0)$ in one time step. Also, in both cases the system is back to set~\eqref{eq:hypotheses}, with $\ub = 0$ (see equations~\eqref{eq:case1} and~\eqref{eq:case2b}). We then need to show that the quantized state cannot keep being in $(0,0)$ indefinitely, but it will eventually switch to $(1,-1)$.
This is indeed the case because according to equation~\eqref{eq:case1}, the system keeps being in~\eqref{eq:hypotheses} with $\ub=0$ and keeps integrating the rounding error until $e$ (necessarily) exceed $1/2$. Then, we are in case 2 since $\round{e} = 1$, and the quantized state switches to $(1,-1)$.
\end{proof}

\begin{proposition} \label{th:uAlgebraic}
Let $1 < \alpha < \frac{3}{2}$, and consider the switched control system described by~\eqref{eq:modeUguale0_delta} and~\eqref{eq:modeDiverso0_delta}. If, at some time $k$, the state satisfies~\eqref{eq:hypotheses}, then, for all the time steps $k+h$, $h>1$:
\begin{align}
&e(k+h) = e(k+h-1) + \round{\ub(k+h-1)} + \Delta_d \label{eq:eProposition}\\
&\ub(k+h) = -\alpha\round{e(k+h)}
\label{eq:uAlgebraic}
\end{align}
\end{proposition}
\begin{proof}
Equation~\eqref{eq:eProposition} follows immediately from the system dynamics in~\eqref{eq:modeUguale0_delta}-\eqref{eq:modeDiverso0_delta}. Based on the proof of Theorem~\ref{th:qinv},~\eqref{eq:uAlgebraic} is trivially satisfied when $\Delta_d = 0$ since in this case $\round{e(k)}=0$, and the system evolves according to~\eqref{eq:caseDelta0}. Let $\Delta_d \neq 0$. If $\round{e(k+1)} = 0$, then $\ub(k+1) = 0$ (see equation~\eqref{eq:case1}). If instead, $\round{e(k+1)} = \sign{\Delta_d}$, then $\ub(k+1) = \ub(k) - \alpha\sign{\Delta_d}$, and in one time step $\ub(k+2) = 0$ (see equations~\eqref{eq:case2} and~\eqref{eq:case2b}).

After time $k+2$, $\ub$ keeps its value to $0$, when $\round{e} = 0$. It become $-\alpha\sign{\Delta_d}$ as soon as $\round{e} = \sign{\Delta_d}$, and then gets back to $\ub = 0$ in one time step.
As a consequence, it is possible to express $\ub(k+h)$, with $h>1$, as:
\begin{align*}
\ub(k+h) = -\alpha \round{e(k+h)},
\end{align*}
thus concluding the proof.
\end{proof}

A possible evolution of the system is shown in Figure~\ref{fig:limitCycle_example}, for $\alpha = 1.1$, $\Delta_d=0.4$, when the switched control system~\eqref{eq:modeUguale0_delta} and~\eqref{eq:modeDiverso0_delta} is initialized at $e(0) = 0.2$, and $\ub(0) = 0.6$. The green square in the figure indicates the initial condition, while the red area indicates the region~\eqref{eq:hypotheses}. The top graph in Figure~\ref{fig:limitCycle_example} shows the phase plot of the system. After the state enters the red area, it ends up in the invariant set characterized in Theorem~\ref{th:qinv}. The central and bottom graphs represent the time evolution of the state variables $e$ and $\ub$ and of their quantized version.
\begin{figure}[tb]
\centering
\includegraphics[scale=1]{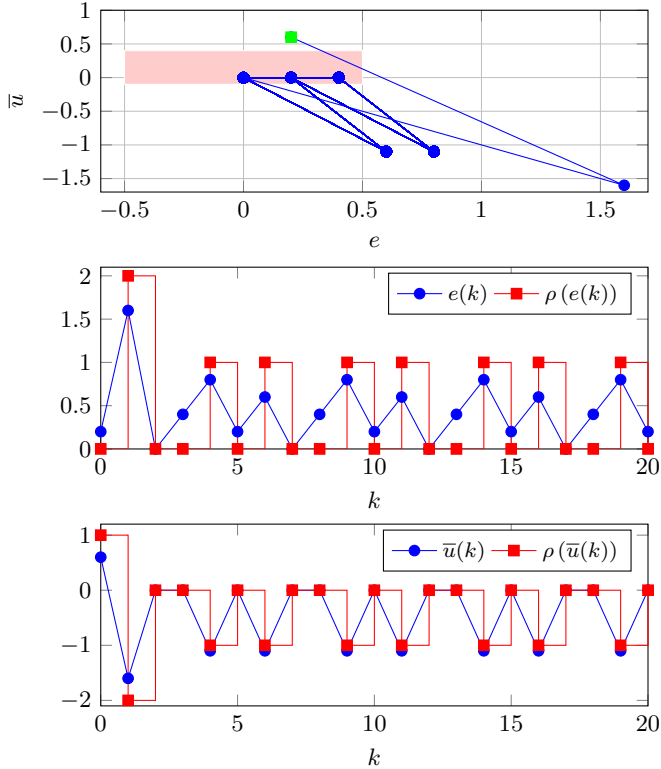}
\caption{Example of a trajectory entering the invariant set characterized in Theorem~\ref{th:qinv}. The top graph shows the phase plot in the state variables $e$ and $\ub$. The lower plots show the time evolution of the state variables and of their quantized versions.}
\label{fig:limitCycle_example}
\end{figure}

Theorem~\ref{th:qinv} provides conditions under which the system ends up in an invariant set where the quantized state variables $\round{e(k)}$ and $\round{\ub(k)}$ range between the values $0$ and $\sign{\Delta_d}$, and $0$ and $-\sign{\Delta_d}$, respectively, with an excursion of amplitude equal to $1$. However, depending on the value of $\alpha$ and of $\Delta_d$ the system may end up on a different invariant set.
This is studied in the following section.

\section{\rev{Numerical analysis of reachability and global attractiveness}}
\label{sec:Reachability}

\rev{The purpose of this section is to study the global attractiveness of the invariant set identified in Theorem~\ref{th:qinv}. To this end, we exploit the fact that once the system has entered the region~\eqref{eq:hypotheses}, in one step it ends up in the invariant set. 
Therefore, we only need to study the reachability of region~\eqref{eq:hypotheses}. 
Providing an analytical reachability analysis for the considered system is quite involved and far from being trivial, due to the quantization effect. In addition, most of the available tools for performing such an analysis (e.g., SpaceEx~\cite{Frehse2011}, Flow*~\cite{Chen2013}, KeYmaera~\cite{Platzer2008}, or Ariadne~\cite{Collins2012}) are meant for continuous time dynamical systems~\cite{Chen2015}.}

\rev{This analysis is parametric in the $(\alpha, \Delta_d)$ pair. 
To carry it out numerically, $\alpha$ and $\Delta_d$ were made variable in the sets $[1.001,1.499]$  and $[-0.5,0.5]$ taking $500$ and $1000$ equally spaced values, respectively.
For each considered pair $(\alpha, \Delta_d)$, system~\eqref{eq:modeUguale0_delta}-\eqref{eq:modeDiverso0_delta} was initialized with $(e(0),\ub(0)) \in [-10,10]^2$, taking $1000$ equally spaced values per coordinate. Note that $[-10,10]^2$ can be taken as representative of the whole state space because for larger values of $(e,\ub)$ the quantization errors become negligible. Outside that set one can therefore assume the system to behave linearly, causing any trajectory to end up in the set itself.}

\rev{The region delimited by the closed curve in Figure~\ref{fig:alphaDeltad} includes all pairs $(\alpha, \Delta_d)$ in the grid for which all the considered initial conditions cause the trajectory to end up in region~\eqref{eq:hypotheses}, and therefore in the invariant set identified in Theorem~\ref{th:qinv}. Note that the values $\Delta_d=\pm 0.5$ are not included in that region.}

\begin{figure}[tb]
\centering
\includegraphics[scale=1]{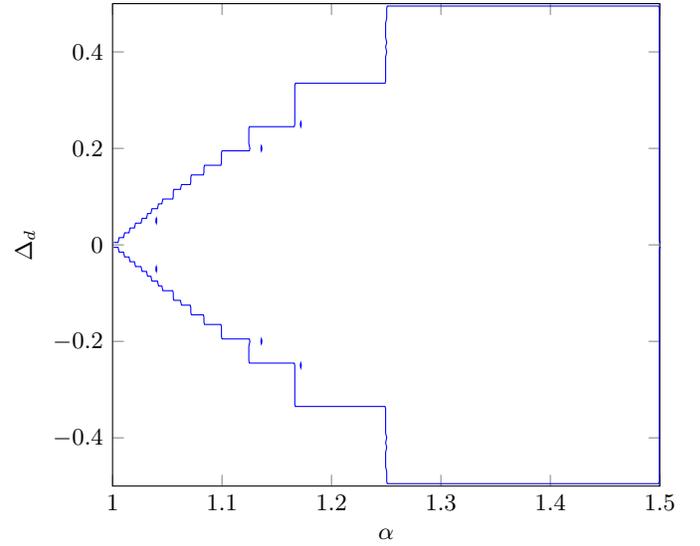}
\caption{The region delimited by the closed curve is the set of $(\alpha,\Delta_d)$ grid couples for which the invariant set of Theorem~\ref{th:qinv} is a global attractor.}
\label{fig:alphaDeltad}
\end{figure}

\rev{This leads to the following statement, which is not a theorem since it is based on a numerical analysis, not on a formal proof.}

\begin{statement}\label{th:globalAttractiveness}
If $5/4 < \alpha < 3/2$ and $|\Delta_d|<0.5$, the invariant set in Theorem~\ref{th:qinv} is globally attractive.
\end{statement}

\rev{In the case when $|\Delta_d|=0.5$, the numerical analysis revealed the existence of an invariant set where the excursion in amplitude of the quantized error is equal to $2$. 
In particular, for $\Delta_d=-0.5$ we get
\begin{align}
\left(\round{e},\round{\ub}\right) \in
\{(-1,2),(1,-1)\},
\label{eq:long_limitCycle1}
\end{align}
whereas for $\Delta_d = 0.5$
\begin{align}
\left(\round{e},\round{\ub}\right) \in
\{(-1,1),(1,-2)\}.
\label{eq:long_limitCycle2}
\end{align}
The invariant sets~\eqref{eq:long_limitCycle1} and~\eqref{eq:long_limitCycle2} can be reached only from a subset of initial conditions, since Theorem~\ref{th:qinv} holds for any $\Delta_d$.}

\rev{It is worth stressing that invariant sets with amplitude $2$ for the quantized error excursion only appeared when $|\Delta_d| = 0.5$. An example is shown in Figure~\ref{fig:LimitCycle_amplitude_2}.}

\rev{For $|\Delta_d| \neq 0.5$, if $\alpha<5/4$ our numerical study showed the existence of two invariant sets, both with unitary excursion amplitude, one of them  being that in Theorem~\ref{th:qinv}.}

\begin{figure}[t]
\centering
\includegraphics[scale=1]{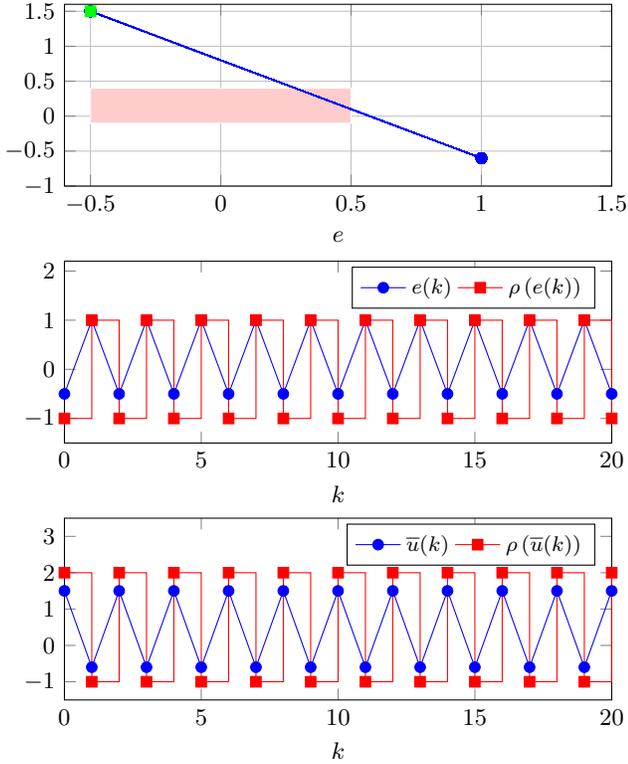}
\caption{Example of an invariant set that can be obtained with the proposed switched scheme when $|\Delta_d| = 0.5$. The state trajectory ends up in an invariant set with excursion of amplitude $2$ for the quantized state $e$. Both the phase plot (top graph) and the time evolution of the state variables $e$ and $\ub$ with their quantized versions (lower plots) are reported.  The red area indicated in the figure is the set~\eqref{eq:hypotheses}.}
\label{fig:LimitCycle_amplitude_2}
\end{figure}

\rev{Figure~\ref{fig:otherLimitCycle_example} shows an example of an invariant set that is different from the one in Theorem~\ref{th:qinv} (but still has a quantized error excursion of amplitude 1). Such an invariant set
\begin{align*}
\left(\round{e},\round{\ub}\right) \in \{(0,1),(1,0)\},
\end{align*}
is obtained for $\alpha = 1.1 (<5/4)$, $\Delta_d=-0.3$, when the system~\eqref{eq:modeUguale0_delta} and~\eqref{eq:modeDiverso0_delta} is initialized at $e(0) = -0.2$, and $\ub(0) = 0.6$. Note that the non-quantized control input behavior shown in Figure~\ref{fig:otherLimitCycle_example} is not easy to predict. On the contrary, the non-quantized control input behavior for the invariant set in Theorem~\ref{th:qinv} can be easily predicted based on $\alpha$ (see Proposition \ref{th:uAlgebraic}).}

Since $\alpha$ is a design parameter, we can choose it so as to enforce the presence only of the invariant set that is fully characterized in Theorem~\ref{th:qinv}, for all disturbances except for those with $|\Delta_d| = 0.5$.

\begin{figure}[t]
\centering
\includegraphics[scale=1]{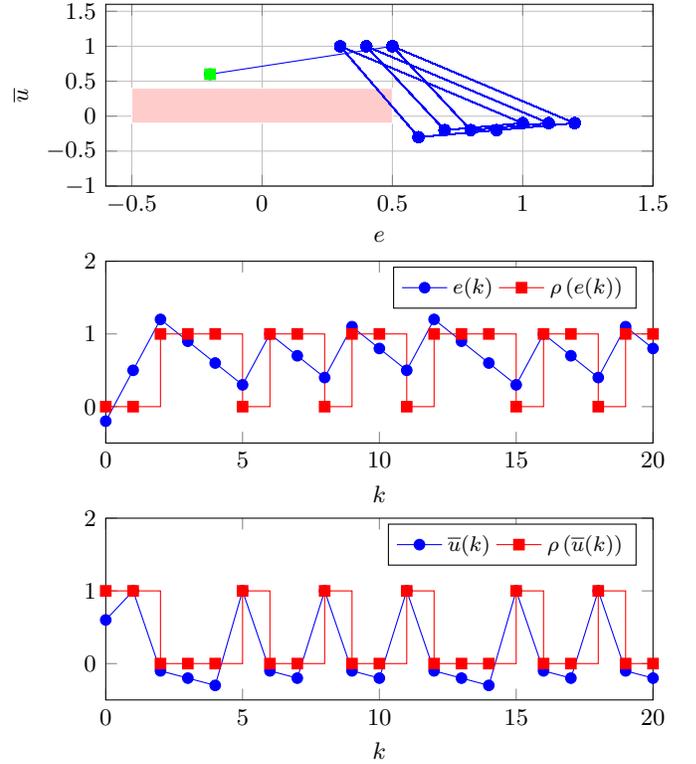}
\caption{Example of an invariant set that can be obtained with the proposed switched scheme for $\alpha <5/4$. The state trajectory ends up in an invariant set with excursion of amplitude $1$ for the quantized state $e$ and $\ub$. Both the phase plot (top graph) and the time evolution of the state variables $e$ and $\ub$ with their quantized versions (lower plots) are reported. The red area indicated in the figure is the set~\eqref{eq:hypotheses}. }
\label{fig:otherLimitCycle_example}
\end{figure}

 \section{Limit cycle analysis}
\label{sec:limitCycles}

\rev{In this section, we analyze the evolution of the switched control system within the invariant set in Theorem~\ref{th:qinv}, and determine possible periodic solutions for the error $e$ and the control input $\ub$, jointly with their period $p$. In particular, we show in Theorem~\ref{th:periodicity} that a necessary and sufficient condition for the presence of periodic solutions is that the disturbance rounding error, hence the disturbance, is a rational number. When dealing with applications in the computing systems domain, rational disturbances can  indeed occur due to the inherently discrete nature of the signals and processes involved. 
Note also that Theorem~\ref{th:periodicity} provides a  necessary and sufficient condition for the existence of a periodic solution so that we can state that for any irrational disturbance, no periodic solution exists, thus further characterizing the behavior of the switched control system.}

We can now start the analysis  by defining the notion of $n$-periodic limit cycle of period $p$.

\begin{definition}[$n$-periodic limit cycle of period $p$]
An \emph{$n$-periodic limit cycle of period $p$}, with $n,p \in \mathbb{N}$, is a solution of the switched control system~\eqref{eq:modeUguale0_delta}-\eqref{eq:modeDiverso0_delta} such that
\begin{align*}
\begin{cases}
e(k+p) = e(k)\\
\ub(k+p) = \ub(k)
\end{cases}, \quad \forall k \geq \overline{k}
\end{align*}
for some $\overline{k} \geq 0$, and the quantized state $(\round{e},\round{\ub})$ switches $n$ times per period.
\end{definition}

\begin{theorem}\label{th:periodicity}
A necessary and sufficient condition for the switched control system to evolve according to an $n$-periodic limit cycle of period $m$ within the invariant set in Theorem~\ref{th:qinv} is that the disturbance rounding error is rational and satisfies
\begin{align*}
|\Delta_d| =\frac{n}{m}, \quad \text{with } 1 \leq n < m, \text{ and } n, m \in \mathbb{N}.
\end{align*}
\end{theorem}
\begin{proof}
Note that when the system is within the invariant set of Theorem~\ref{th:qinv}, the algebraic relation~\eqref{eq:uAlgebraic} holds. Therefore, we just need to show that the state variable $e$ evolves on the $n$-periodic limit cycle of period $m$.

We start by showing that a necessary condition for this to hold is that $|\Delta_d|$ is rational.

Suppose that at a certain time step $h$ the system is within the (minimal) invariant set of Theorem~\ref{th:qinv}. Assume also, without loss of generality, that $\left(\round{e(h)},\round{\ub(h)}\right) = (0,0)$. This entails that $|e(h)|<0.5$ and that the input $\round{u(h)}+ \overline{d}$ to the process is equal to $\Delta_d$ since $\round{u(h)}=-\round{\overline{d}}$ from equation~\eqref{eq:utilde}.
Indeed, the input to the process keeps constant and equal to $\Delta_d$ for $k$ time steps, until $|e(h+k)|$ exceeds or gets equal to $0.5$ if $\Delta_d>0$, $-0.5$ if $\Delta_d <0$. At time $h+k$, then,  $\round{e(h+k)}\neq 0$ and  the pair $(\round{e(h+k)},\round{\ub(h+k)})$ switches to $(\sign{\Delta_d},-\sign{\Delta_d})$ in the invariant set. The number of steps $k$ is given by the following formula
\begin{equation}
\label{eqn:lambda-equation}
k=\lambda(\Delta_d,x^{+(0)}):=\left\lceil \frac{0.5 \sign{\Delta_d}-x^{+(0)}}{\Delta_d} \right\rceil,
\end{equation}
where we set $e(h)=x^{+(0)}$.
Observe that $\lambda(\Delta_d,x^{+(0)})$ approaches infinity as $\Delta_d$ tends to zero, in accordance with Theorem~\ref{th:qinv} where the invariant set is composed only of the value 0 if $\Delta_d=0$.

The value $x^{+(1)}$ taken by $e(h+k+1)$ can be obtained as
\begin{align}
\label{eq:invariant-evolution}
x^{+(1)} = x^{+(0)} + \lambda(\Delta_d,x^{+(0)})\Delta_d + \Delta_d - \sign{\Delta_d},
\end{align}
since the process integrates an input that is constant and equal to $\Delta_d$ for $k=\lambda(\Delta_d,x^{+(0)})$ steps, and then receives as input $\round{u(h+k)}+\overline{d}=\round{\ub(h+k)}-\round{\overline{d}}+\overline{d}=-\sign{\Delta_d}+\Delta_d$ at time $h+k$.

If $x^{+(1)}$ is equal to $x^{+(0)}$, then the evolution of state $e$ of the system is periodic with period $\lambda(\Delta_d,x^{+(0)})+1$, and we have an $1$-periodic limit cycle of period $k+1$, because one single switch is needed within the invariant set to reset the state of the process to its original value, and this required $k+1$ steps.
If $x^{+(1)}\ne x^{+(0)}$, we can further iterate the same reasoning by considering $i>1$ switches within the invariant set and computing $x^{+(i)}$, $i>1$.
If there exists some integer $N>1$ such that $x^{+(N+h)}=x^{+(h)}$, for some $h\ge 0$, then, the state of the process evolves according to an $N$-periodic limit cycle.

More specifically, we need to compute
\begin{align*}
x^{+(N+h)} &= x^{+(h)} +\\
           &~+ \sum_{i=0}^{N-1}{\lambda(\Delta_d,x^{+(i+h)})\Delta_d} + N( \Delta_d - \sign{\Delta_d}),
\end{align*}
and set $x^{+(N+h)}=x^{+(h)}$, which reduces to solving
\begin{align*}
\left(\sum_{i=0}^{N-1}{\lambda(\Delta_d,x^{+(i+h)})} + N\right) |\Delta_d|=N.
\end{align*}
For this equation to admit a solution we must have
\[
|\Delta_d|=\frac{N}{L},
 \]
where we set $L= \left(\sum_{i=0}^{N-1}{\lambda(\Delta_d,x^{+(i+h)})} + N\right)$.
Note that since $L$ is an integer larger than $N$, for a periodic trajectory of the state process $e$ to exist, the absolute value of disturbance  quantization error $|\Delta_d|$ must be a rational number of the form $\frac{n}{m}$ with $n<m$. Irrational values for $|\Delta_d|$ are then incompatible with periodic solutions.

We now show that the condition $|\Delta_d| = \frac{n}{m}$ being a rational number is sufficient to have an $n$-periodic limit cycle of period $m$.

Observe that by definition of $\lambda $ as the minimum number of steps needed for $\round{e(h+k)}\neq 0$ starting from $e(h)=x^{+(0)}$, we have that
\begin{align*}
e(h+k)=&x^{+(0)} + \lambda(\Delta_d,x^{+(0)})\Delta_d \\
&\in
\begin{cases}
[0.5, 0.5 +\Delta_d) & \Delta_d>0\\
(-0.5+\Delta_d, -0.5] & \Delta_d<0
\end{cases}.
\end{align*}
This entails that
$x^{+(1)}$ in \eqref{eq:invariant-evolution} satisfies
\begin{align*}
x^{+(1)} \in
\begin{cases}
[-0.5+\Delta_d, -0.5 +2\Delta_d) & \Delta_d>0\\
(0.5+2\Delta_d, 0.5 +\Delta_d] & \Delta_d<0
\end{cases}
\end{align*}
irrespectively of $x^{+(0)}$. And this hold true for every $x^{+(i)}$ value of $e$ after $i$ switches within the invariant set, with $i\ge 1$.

Let $|\Delta_d|=\frac{n}{m}$, where $n$ and $m$ are coprime integers, $m > n \ge 1$,we next show that, after at least one switch has occurred within the invariant set, then, the switched control system starts evolving according to an $n$-periodic limit cycle of period $m$. We refer to the case when $\Delta_d>0$. The same reasoning applies to  $\Delta_d<0$.

If there were no further switches after time $h+k$ when $e(h+k)= x^{+(1)}$, then, $e(h+k+m)$ would take values in $[x^{+(1)},\  x^{+(1)}+m\Delta_d]=[x^{+(1)}, \ x^{+(1)}+n]$ since the system would integrate a constant input equal to  $\Delta_d$ for $m$ steps. However, as soon as $e$ becomes larger than or equal to the threshold $0.5$, then, its value is decreased by $1$, so that if there were exactly $n$ switches in the time frame $[h+k, \, h+k+m]$, then, $e(h+k+m)=x^{+(1)}=e(h+k)$ and a periodic solution would be in place.   Now, in order to show that there are exactly $n$ switches in the time frame $[h+k, \ h+k+m]$, one should simply check that $[x^{+(1)}, \ x^{+(1)}+n]$ contains $\{0.5+i, i=0,1, \dots, n-1\}$ and does not contain $0.5+n$.\\
Clearly, $0.5+i$ is contained in $[x^{+(1)},\ x^{+(1)}+n]$ for $i=0$ and $i=n-1$, since $x^{+(1)}>-0.5$.
Now we need to show that $x^{+(1)}+n<0.5+n$ to conclude that $[x^{+(1)}, \ x^{+(1)}+n]$ does not contain $0.5+n$. Indeed,  since $x^{+(1)}<-0.5 +2\Delta_d$, we have that $x^{+(1)}+n<n+2\Delta_d-0.5$, which entails $x^{+(1)}+n<n+0.5$ given that $\Delta_d \le 0.5$.

This concludes the proof.
\end{proof}

Figure~\ref{fig:aperiodic} plots the evolution of the state of the control system for $\Delta_d=\sqrt{2}/3$, $\alpha=1.1$, $e(0) = 0.2$, and $\ub(0) = 0.6$. Notice that since $\Delta_d$ is irrational, the obtained trajectory is not periodic.
\begin{figure}[t]
\centering
\includegraphics[scale=1]{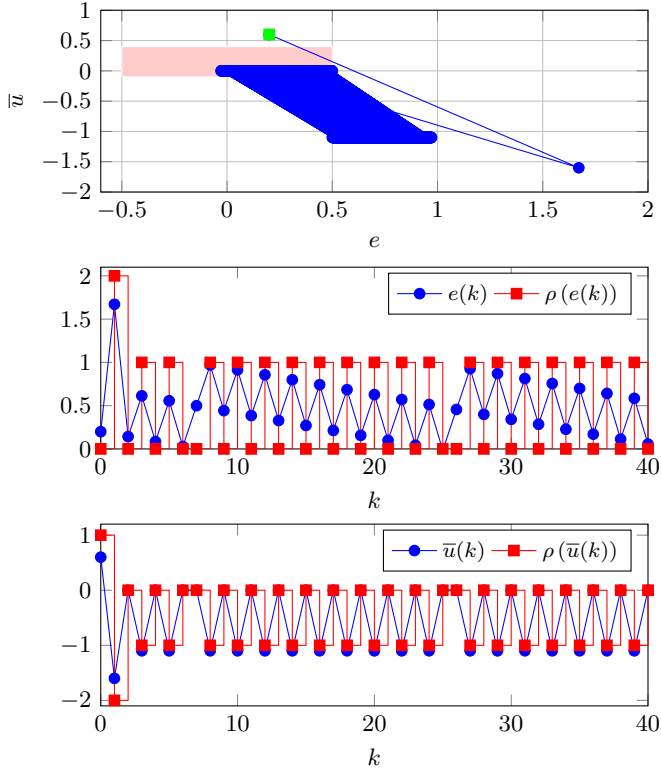}
\caption{Evolution of the switched control system when $\Delta_d = \sqrt{2}/3$.}
\label{fig:aperiodic}
\end{figure}

Figure~\ref{fig:limitcycle-1} shows an example of a $1$-periodic limit cycle of period $5$ obtained for $\Delta_d=0.2=\frac{1}{5}$, starting from the initial condition $e(0) = -0.4$, $\ub(0) = 0.2$.  Figure~\ref{fig:limitcycle-2} shows a $2$-periodic limit cycle of period $5$ for $\Delta_d=-0.4=\frac{2}{5}$ starting from the same initial condition $e(0) = -0.4$, $\ub(0) = 0.2$.

\begin{figure}[h]
\centering
\includegraphics[scale=1]{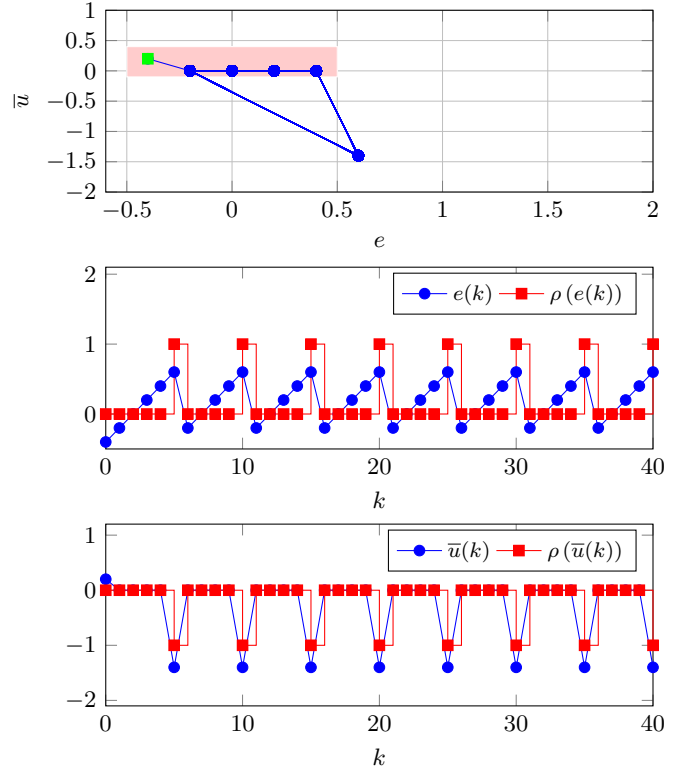}
\caption{Evolution of the switched control system when $\Delta_d = 1/5$.}
\label{fig:limitcycle-1}
\end{figure}

\begin{figure}[h]
\centering
\includegraphics[scale=1]{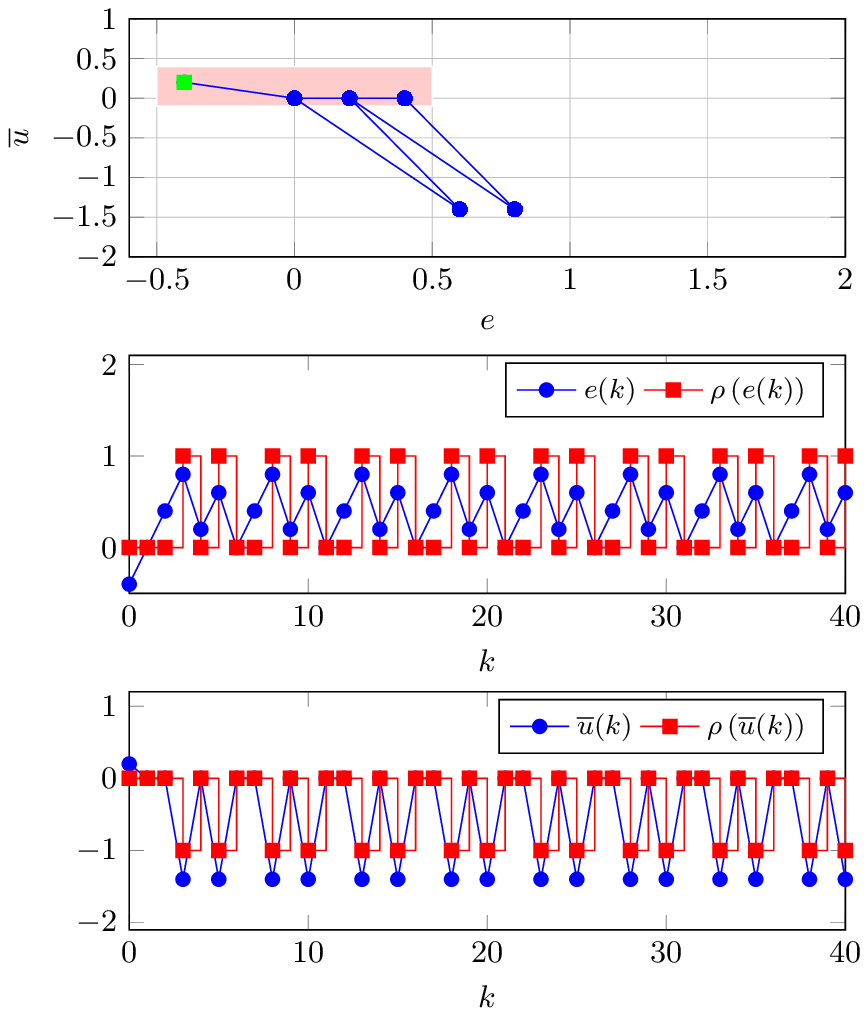}
\caption{Evolution of the switched control system when $\Delta_d = 2/5$.}
\label{fig:limitcycle-2}
\end{figure}

The following corollary directly follows from Theorem~\ref{th:qinv} and Theorem~\ref{th:periodicity}, and summarizes the results of the limit cycle analysis.
\begin{corollary}\label{th:corollary}
If $1 < \alpha < \frac{3}{2}$ and $|\Delta_d| = \frac{n}{m}$, where $n,m\in \mathbb{N}$, $1 \leq n < m$, and $|\Delta_d| < \frac{1}{2}$, then the switched control system~\eqref{eq:modeUguale0_delta}-\eqref{eq:modeDiverso0_delta} admits a limit cycle where the error $e$ is kept within $[-0.5+\Delta_d,\ 0.5 +\Delta_d)$ if $\Delta_d > 0$, and within $(-0.5+\Delta_d,\ 0.5 + \Delta_d]$ if $\Delta_d < 0$ with a corresponding quantized version excursion of $1$.
\end{corollary}
\begin{proof}
We only need to show that $e$ is kept within $[-0.5+\Delta_d,\ 0.5 +\Delta_d)$ if $\Delta_d > 0$, and within $(-0.5+\Delta_d,\ 0.5 + \Delta_d]$ if $\Delta_d < 0$. Suppose that $\Delta_d > 0$. By Theorem~\ref{th:qinv} and Proposition~\ref{th:uAlgebraic}, we have that at some time $k>1$ after entering the invariant set $\round{e(k)} = \round{\ub(k)} = 0$, and $\ub(k) = 0$. Then, the error evolves starting from $|e(k)|<1/2$, according to~\eqref{eq:eProposition} which becomes:
\begin{align}
e(k+h) = e(k+h-1) + \Delta_d
\label{eq:integralDynamics}
\end{align}
(since $\round{e} = \round{\ub} = 0$), until $1/2 \leq e(k+h) < 1/2+\Delta_d$, when $\round{e(k+h)} = 1$ and hence $\ub(k+h) = -\alpha\round{e(k+h)} = -\alpha$. At time $k+h+1$, the error is reset to $e(k+h+1) = e(k+h) +\Delta_d -1$, so that $-1/2 +\Delta_d \leq e(k+h+1) < -1/2+2\Delta_d$, and we are back to the integral dynamics~\eqref{eq:integralDynamics} because $\round{e} = \round{\ub} = 0$. From this analysis it follows that $-1/2+\Delta_d \leq e < 1/2 + \Delta_d$. Analogous derivations can be carried out for the case $\Delta_d <0$.
\end{proof}

\begin{remark}
The reachability numerical analysis in Section~\ref{sec:Reachability} shows that the limit cycle in Corollary~\ref{th:corollary} is globally attractive if we restrict $\alpha$ to the range $\frac{5}{4}< \alpha < \frac{3}{2}$.
\end{remark}

 \section{Simulation results}
\label{sec:Experiments}

We first present some simulation results comparing the three cases when no quantization is present in the control scheme, and when quantization is present and either the PI or its switched extension is implemented. Notice that in the absence of quantization the PI controller and its switched extension coincide. Figure~\ref{fig:differentMethods} reports the simulation runs for the three cases for a finite horizon of $30$ time units. In all three plots the error is normalized, i.e., a unitary resolution is assumed. The value used for $\alpha$ is $11/8$, and $\Delta_d = \sqrt{2}-1$, while the system state is initialized at $e(0) = 0$, and $u(0)= 0$.

While in the absence of quantization the error converges to $0$ with the designed controller, when quantization is in place it is not possible anymore to guaranteeing convergence to zero. In the case of PI control, the error oscillates in the area $[-1,1]$, while in the case of its switched extension, it ends up oscillating in the region $[0,1]$ according to Statement~\ref{th:globalAttractiveness} and Theorem~\ref{th:qinv}. It is worth noticing that for the chosen value of $\Delta_d$ the evolution of the control system state cannot be periodic by Theorem~\ref{th:periodicity}. This is reflected in the evolution of $e$ that oscillates in the gray area, but always assumes different values in the set.

\begin{figure}[h]
\centering
\includegraphics[scale=1]{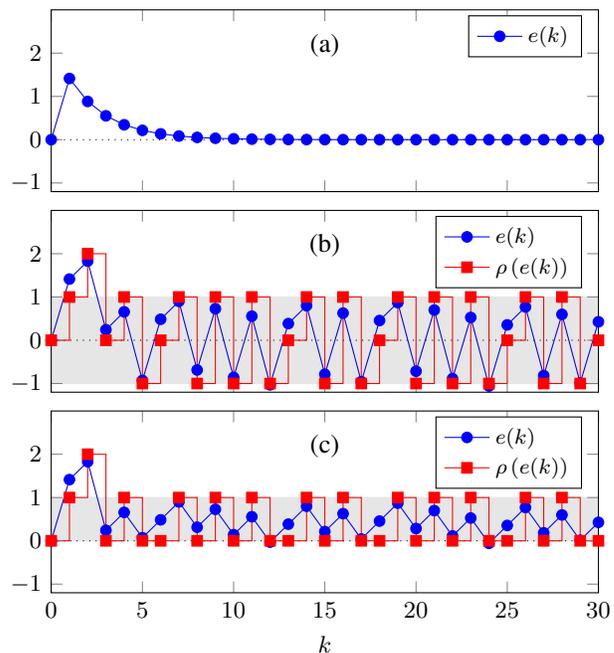}
\caption{Behavior of quantized (red line with squares) and non quantized (blue line with circles) error in a simulation run: (a) without quantizers, (b) with the standard PI with quantizers, and (c) with the switched PI with quantizers.}
\label{fig:differentMethods}
\end{figure}

\rev{We now consider a time-varying disturbance, which is initially constant and takes the value $d = \overline{d}_1 = 2.6$ ($\Delta_d = -0.4 <0$), then, starts decreasing linearly at time $k=20$ till it hits the value $d = \overline{d}_2 = 2.4$ at $k=40$ ($\Delta_d = 0.4>0$), and finally keeps constant.}

The results of the simulation with the switched controller are shown in Figure~\ref{fig:example_explained_switch}, with the error $e$, the control signal $u$, and the disturbance $d$ on the left column, and their quantized versions on the right column. The system is initialized with $e(0) = 0$, $u(0) = 0$, and we set $\alpha = 11/8$.

\begin{figure*}[t]
\centering
\includegraphics[scale=1]{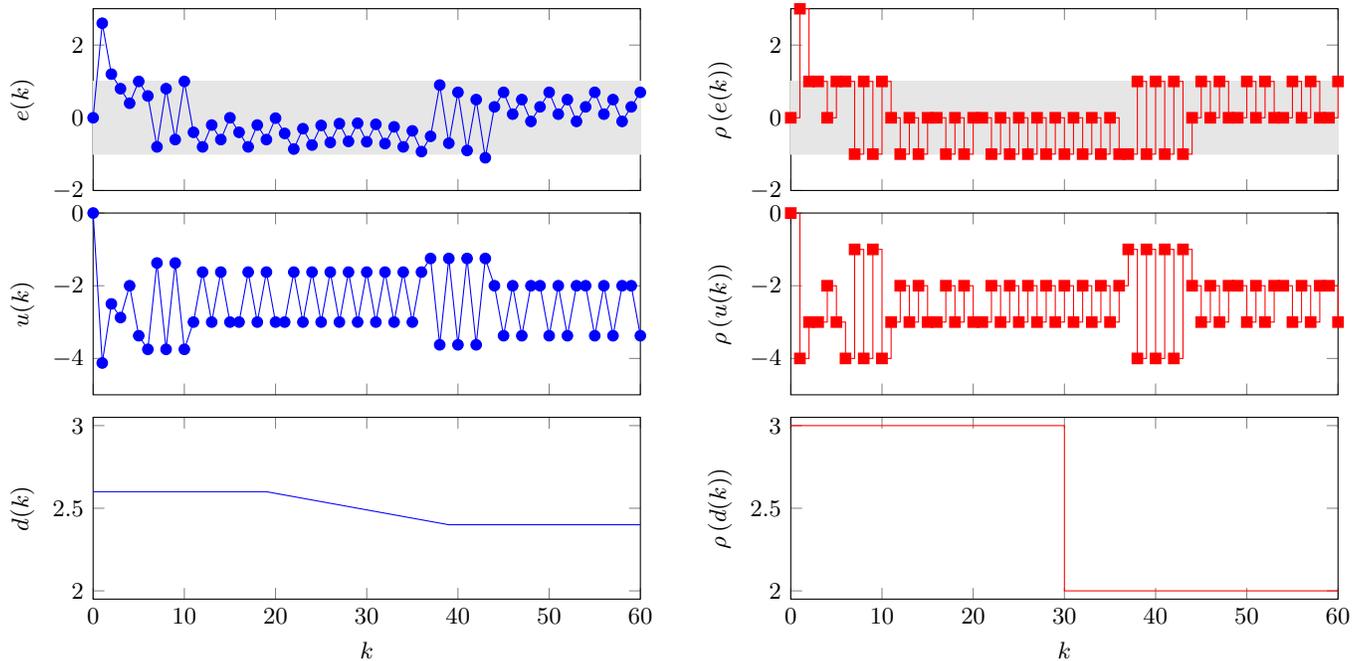}
\caption{Simulation with a time-varying disturbance and the switched PI controller.}
\label{fig:example_explained_switch}
\end{figure*}

\rev{Note that the abrupt change of sign of $\Delta_d$ when the disturbance crosses the threshold $2.5$ at time $k=30$ causes a transient which can be seen from the error behavior, and it is reflected in the quantized version only later, at time $k=37$, when the quantized error starts oscillating between $[-1,1]$ and correspondingly the quantized control input oscillates between $[-4,-1]$. Such oscillations stop when the (new) invariant set described in  Theorem~\ref{th:qinv} is reached according to Statement~\ref{th:globalAttractiveness}. The quantized error then exceeds the minimum resolution only temporarily during the (delayed) transient cause by the threshold crossing.}
In the case of the standard PI controller, the quantized error and the quantized control input keep oscillating between $[-1,1]$ and $[-4,-1]$, respectively, for the whole time horizon, irrespectively of the fact that the disturbance crosses the threshold (see Figure~\ref{fig:example_explained_noswitch}).

If we change $\overline{d}_2$ to $2.501$, the threshold $2.5$ is not crossed by the disturbance and the system keeps evolving in the same invariant set (see Figure~\ref{fig:no_commutation}).

\begin{figure*}[t]
\centering
\includegraphics[scale=1]{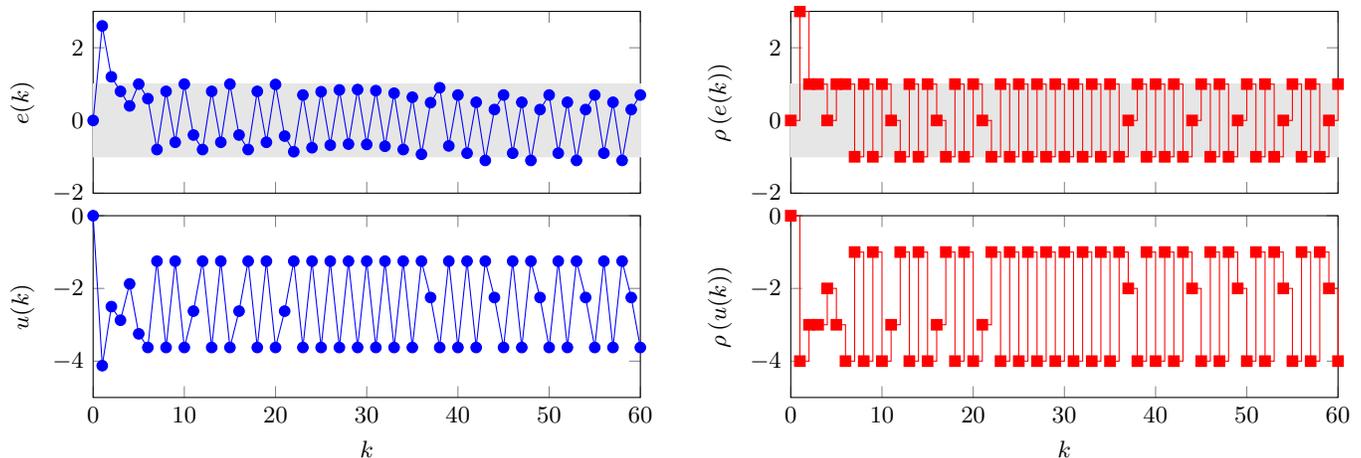}
\caption{Simulation with the time-varying disturbance in Figure~\ref{fig:example_explained_switch} and the standard PI controller.} \label{fig:example_explained_noswitch}
\end{figure*}

\begin{figure*}[t]
\centering
\includegraphics[scale=1]{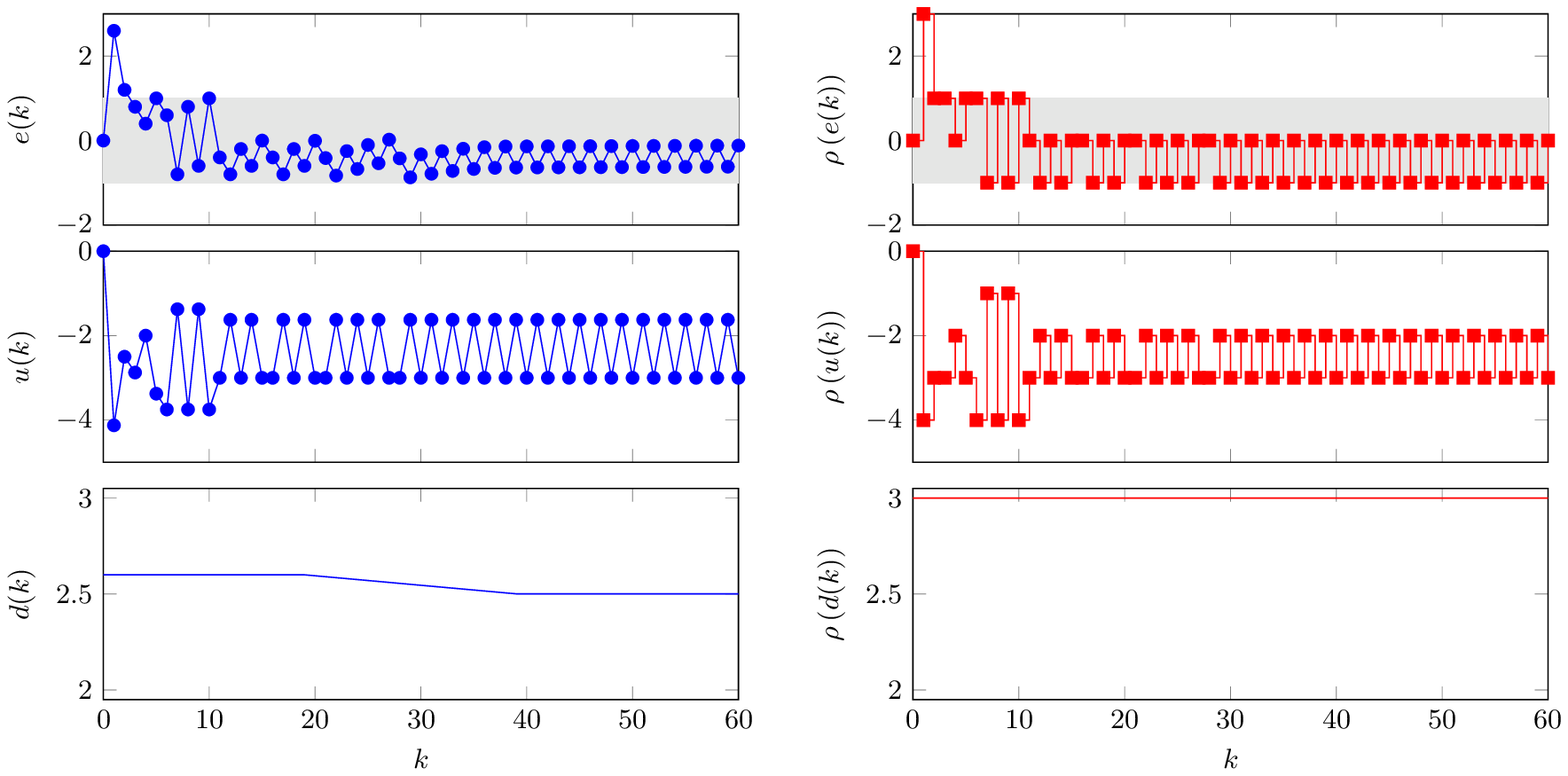}
\caption{Simulation with a time-varying disturbance that has constant quantized value, using the switched PI controller.} \label{fig:no_commutation}
\end{figure*}

The results presented next refer to a simulation campaign aimed at investigating the effect of the disturbance magnitude on the control performance, with and without the proposed switched extension.

The campaign was carried out by choosing the values of $\overline{d}$ reported in Table~\ref{tab:eRMS}. For each value of $\overline{d}$, two models -- one with bare PI control  and the other with switched PI -- were initialized to $e(0) = 0$ and $u(0) = 0$, and then subjected to a constant disturbance of the selected amplitude. Data were collected from the two simulated experiments just described over a finite horizon of $H = 1000$ time units.
We assess performance by computing the Root Mean Square (RMS) value of the quantized error, that is defined as:
\begin{align*}
RMS_{\round{e}} = \sqrt{\dfrac{1}{H}\sum_{i=0}^{H-1} \round{e(i)}^2}
\end{align*}
where $H$ is the length of the simulation.

Table~\ref{tab:eRMS} summarizes the results and shows that the proposed switched scheme decreases the $RMS_{\round{e}}$ by $30\%$.

\begin{table}[htb]
 \centering
 \begin{tabular}{ccc}
  \hline
  disturbance & \multicolumn{2}{c}{$RMS$ performance index}\\
  $\overline{d}$ &
  standard PI &
  switched PI  \\
  \hline\hline
  $\pm 0.01$       & $0.138$  & $0.100$ \\
  $\pm 0.02$       & $0.197$  & $0.141$ \\
  $\pm 0.04$       & $0.281$  & $0.200$ \\
  $\pm 0.05$       & $0.314$  & $0.223$ \\
  $\pm 0.1$        & $0.446$  & $0.316$ \\
  $\pm 0.2$        & $0.631$  & $0.447$ \\
  $\pm 0.4$        & $0.893$  & $0.632$ \\
  $\pm (\sqrt{2}-1)$ & $0.909$  & $0.643$ \\
  \hline
 \end{tabular}
 \caption{$RMS$ performance index of the simulation campaign.}
 \label{tab:eRMS}
\end{table}
 \section{Conclusions and future work}
\label{sec:Conclusions}

A switched control scheme was proposed for reducing the degradation effect due to the quantization of both control and controlled variables in a system described as an integrator with unit delay. Set invariance and limit cycle analysis were performed, jointly with a numerical reachability study, to assess the switched control scheme performance and provide guidelines for control tuning. In particular, necessary and sufficient conditions for the presence of $n$-periodic limit cycles of period $p$ were discussed. Finally, simulation results confirm the effectiveness of the proposed solution.

Future work will concern the evaluation of the proposed approach in specific types of applications, where the quantization effect is relevant.
Results are confined to a specific class of systems. Further investigations are needed also to extend the proposed approach to a larger class of problems.

\bibliographystyle{IEEEtran}
\bibliography{2017-TAC-quantized}

\end{document}